\documentclass[conference,letterpaper]{IEEEtran}

\usepackage[utf8]{inputenc} 
\usepackage[T1]{fontenc}
\usepackage{url}
\usepackage{ifthen}
\usepackage{cite}
\usepackage[cmex10]{amsmath} 
\usepackage{graphicx} 
\usepackage{amssymb}
\usepackage{latexsym}
\usepackage{amsmath}
\usepackage{amsthm}
\usepackage{empheq}
\usepackage{bm}
\usepackage{booktabs}
\usepackage[dvipsnames]{xcolor}
\usepackage{pagecolor}
\usepackage{caption}
\usepackage{subcaption}
\usepackage{tikz,pgfplots}
\usetikzlibrary{matrix,decorations.pathreplacing,positioning}
\usepackage{enumitem}
\usepackage{wasysym}
\usepackage{mdframed}
\usepackage{ragged2e}
\usepackage{multirow}
\usepackage{array}
\usepackage{hhline}
\usepackage{enumitem}

\usepackage{arydshln}

\usepackage{centernot}
\usepackage{mathtools}
\usepackage{stmaryrd}
\usepackage{relsize}

\numberwithin{equation}{section}

\newtheorem{theorem}{Theorem}[section]

\newtheorem{proposition}[theorem]{Proposition}
\newtheorem{lemma}[theorem]{Lemma}

\theoremstyle{definition}
\newtheorem{definition}[theorem]{Definition}
\newtheorem{example}[theorem]{Example}

\newtheorem{remark}[theorem]{Remark}

\newtheorem{construction}[theorem]{Construction}

\newcommand{\ket}[1]{\lvert#1\rangle}

\newcommand{\F}{\mathbb{F}}

\newcommand{\rmv}[1]{}

\usepackage{xcolor}
\usepackage{soul}

\begin{document}
\title{List-Decodable Folded Quantum Hermitian Codes}

\author{%
Julia Shapiro\IEEEauthorrefmark{1} and \IEEEauthorblockN{Gretchen L. Matthews\IEEEauthorrefmark{2}}
 \IEEEauthorblockA{
 Department of Mathematics \\
                    Virginia Tech\\
                    Blacksburg, VA 24061 USA\\
                    Email: \{gmatthews, juliams22\}@vt.edu}}

\maketitle

\renewcommand\thefootnote{}
\footnotetext{
\IEEEauthorrefmark{1}~J. Shapiro is supported by the Department of War Cyber Service Academy Scholarship.\\
\IEEEauthorrefmark{2}~G. L. Matthews is partially supported by NSF DMS-2502705,  NSF CNS-2413218,  and the Commonwealth Cyber Initiative.

}

\begin{abstract} Folded Reed-Solomon codes, introduced by Guruswami and Rudra in 2007, have been shown to achieve the information-theoretically best possible trade-off between the rate of a code and the error-correction radius. In 2024, Bergamaschi, Golowich and Gunn extended this framework by constructing folded quantum Reed-Solomon codes (CSS codes obtained by folding) demonstrating that these codes tolerate errors up to the quantum Singleton bound. In this paper, we construct folded quantum Hermitian codes using the CSS framework and show that these codes are also list-decodable, tolerating errors up to the quantum Singleton bound. Compared to Reed–Solomon codes, Hermitian codes admit comparable lengths over smaller alphabets, enabling more efficient implementations. 
\end{abstract}

\section{Introduction}

Quantum error‑correcting codes derived from algebraic–geometric (AG) codes were first studied  by Ashikhmin, Litsyn, and Tsfasman~\cite{AshikhminLitsynTsfasman2001}, followed by numerous works developing explicit AG‑based  constructions and refinements~\cite{AshikhminKnill2001,MatsumotoUyematsu2007,EzermanKirov2018,Kim2022,KimMatthews2008,KimWalker2008,PereiraEtAl2021,HernandoMcGuireMonserratMoyano2020}. Hermitian curves have played a particularly prominent role in this literature, due to their rich automorphism groups and favorable asymptotic parameters~\cite{SarvepalliKlappenecker2009,LaGuardia2011,ChristensenGolik2020,CaoEtAl2020}.
In parallel, list decoding has been a central topic in coding theory since the work of Sudan and Guruswami–Sudan~\cite{SUDAN1997180,743426}, with subsequent advances improving upon the Johnson decoding radius~\cite{Johnson1962ANU,1530722}. A major breakthrough was the introduction of folded Reed–Solomon codes~\cite{10.1145/1132516.1132518,4418476}, in which consecutive code symbols are grouped over a larger alphabet, enabling efficient multivariate interpolation and fundamentally improving the rate–error‑radius trade‑off. This paradigm permits list decoding up to a $1-R-\epsilon$ fraction of errors, approaching the Singleton bound; folded Reed–Solomon codes can be viewed as a special case of Parvaresh–Vardy codes~\cite{1530722}. These ideas were later extended to folded algebraic–geometric codes, yielding constructions with polynomial list size over constant or bounded alphabets and decoding radii approaching capacity~\cite{10.5555/2634074.2634208,10.1145/3506668}.

Folded quantum codes arise by translating folding techniques into the quantum setting. Motivated by these classical successes and developments in quantum CSS and AG codes, researchers explored whether folding could similarly enhance quantum error correction, particularly in regimes beyond unique decoding. A first treatment appeared in work by Bergamaschi, Golowich, and Gunn~\cite{bergamaschi2024approaching}, who constructed folded quantum Reed–Solomon codes and showed that they achieve the quantum Singleton bound. In particular, they proved list decodability up to radius $\frac{1-R}{2}-\epsilon$, where $R$ is the rate of the folded quantum Reed–Solomon code as a CSS code. As in the classical case, however, these codes have large alphabet sizes; to address this, the authors introduced quantum distance amplification to reduce the alphabet size.
Roughly speaking, folded quantum codes encode logical qubits into CSS or stabilizer codes derived from folded classical codes, where multiple consecutive symbols are bundled into a higher‑dimensional qudit prior to quantization. Starting from a classical AG code, folding groups consecutive evaluations into a single symbol over an extension field, preserving rate while introducing stronger global algebraic structure. In the quantum setting, these folded codes are combined via the CSS construction so that $X$‑ and $Z$‑checks inherit this structure. Errors affecting many physical qubits thus appear as fewer block errors at the folded level, enabling quantum list‑decoding procedures to identify a small list of candidate logical states even beyond the unique‑decoding radius~\cite{LG}. Subsequent verification or purification then recovers the correct logical state, allowing recovery of the encoded quantum information. This mechanism enables folded quantum codes to approach the quantum Singleton bound while retaining efficient decoding and moderate alphabet size.
In this paper, we introduce quantum codes from folded algebraic‑geometry codes using positive‑genus curves, extending the Reed–Solomon construction based on the genus‑$0$ projective line. We give explicit constructions of folded quantum Hermitian codes and general folded quantum AG codes, and show that codes from positive‑genus curves can achieve smaller alphabet sizes than folded Reed–Solomon codes with distance amplification, while preserving comparable decoding guarantees.
The paper is organized as follows. Section~\ref{prelim} introduces notation and background. In Section~\ref{results}, we construct folded quantum Hermitian and general folded quantum AG codes, establish their list‑decoding capacity up to the quantum Singleton bound, and show that distance amplification is unnecessary. We also present an entanglement‑assisted construction of folded quantum AG codes. The paper concludes with a summary in Section~\ref{conclusion}.

\section{Preliminaries}\label{prelim}
In this section, we introduce the necessary notation for the remainder of this paper. See \cite{stichtenoth,10.1007/BFb0087995} for details. Throughout,  $\mathbb{F}_q$ denotes the finite field with $q$ elements, where $q = p^s,$ for a prime $p$. An \emph{$[n,k,d]$ code over $\mathbb{F}_q$} is a $k$-dimensional linear subspace
$
C \subseteq \F_q^n$
with minimum (Hamming) distance
$
d := \min\{ \mathrm{wt}(c) \mid c \in C,\ c \neq 0 \}$,
where $\mathrm{wt}(c)$ denotes the number of nonzero coordinates of $c$. Such a code is able to correct any $t:=\left\lfloor \frac{d-1}{2} \right\rfloor$ errors; that is, there is a unique codeword of $C$ contained in the Hamming ball of radius $t$ around a received word $w \in \F_q^n$ that differs from a codeword $c$ in at most $t$ positions. \rmv{When the number of errors exceeds $\left \lfloor \frac{d-1}{2} \right \rfloor$, there may be multiple codewords that are close to the received word.}
List decoding can address more than $t$ errors by relaxing   the uniqueness requirement: instead of returning a single codeword, the decoder outputs a list of all codewords within a specified decoding radius, allowing decoding beyond the unique decoding radius, at the cost of possibly producing multiple candidate codewords.

For a received word $y \in \mathbb{F}_q^n$ and a real number $\tau \in [0,1]$, we say that a codeword
$c \in \mathcal{C}$ is \emph{within relative distance $\tau$ of $y$} if
\[
\frac{1}{n} \, d_H(c,y) \le \tau,
\]
where $d_H(\cdot,\cdot)$ denotes Hamming distance.
Equivalently, $c$ differs from $y$ in at most $\tau n$ positions. A code $\mathcal{C} \subseteq \mathbb{F}_q^n$ is  $(\tau,L)$-list-decodable if for every
received word $y \in \mathbb{F}_q^n$, the number of codewords within relative distance $\tau$ of $y$ is at most $L$.
Formally,
\[
\left| \left\{ c \in \mathcal{C} \;:\; d_H(c,y) \le \tau n \right\} \right| \le L
\quad \text{for all } y \in \mathbb{F}_q^n.
\]
The parameter $\tau$ is the (relative) decoding radius, and $L$ is an upper bound on the output list size.
Smaller $L$ corresponds to stronger guarantees on ambiguity, while larger $\tau$ corresponds to the ability to correct a larger fraction of errors.
Unique decoding is the special case $L=1$ and $\tau < \frac{d}{2n}$. For $\tau \ge \frac{d}{n}$, it is possible for exponentially many codewords to lie within distance $\tau n$ of a received word, so bounding $L$ is essential.

Let $C \subseteq \mathbb{F}_q^n$ be an $[n,k,d]$ code and $m \mid n$.
Partition the coordinate positions $\{1,\dots,n\}$ into $\frac{n}{m}$ consecutive blocks
of size $m$. In particular, the
folded code has length $\frac{n}{m}$ and preserves the rate of the original
code, but $\frac{k}{m}$ need not be an integer unless $m\mid k$.
The folded code of $C$ with folding parameter $m$ is 
\[C^{(m)}:= \left\{ 
\bigl( c^{(1)}, c^{(2)},\ \dots, c^{(\frac{n}{m})} \bigr): c \in C
\right\} \subseteq (\mathbb{F}_q^{\,m})^{\frac{n}{m}}
\]
where $c^{(i)}:=(c_{(i-1)m+1}, \dots, c_{im}) \in \F_q^m$.
The folded code $C^{(m)}$ may be seen as a $\bigl[\frac{n}{m},\;\frac{k}{m},\;\ge\lceil \frac{d}{m} \rceil\bigr]$ code over the alphabet $\mathbb{F}_q^{\,m}$.
Equivalently, $C^{(m)}$ is obtained by viewing each block of $m$
symbols of $C$ as a single symbol over the alphabet $\mathbb{F}_q^{\,m}$. 
Notice that while the length of the $C^{(m)}$ is an integer (as $m\mid n$), the dimension is not required to be. 
As we shall see, when $C$ is an AG code, automorphisms may be used to determine the folding. Before moving to that, we note the dual of a folded code defined as follows. 
\begin{definition}\label{def1}    
The dual of the folded code $C^{(m)}$ is the set 
\[ (C^{(m)})^{\perp} :=
\left\{ 
{\bf{w}} \in (\F_q^m)^N: {\bf{w}} 
\cdot {\bf{c}} = 0 \ 
\forall {\bf{c}} \in C^{(m)} \right\}
\]
where ${\bf{w}}:=
\left(w^{(1)}, \ldots, w^{(\frac{n}{m})} \right)$ and ${\bf{c}}:=
\left(c^{(1)}, \ldots, c^{(\frac{n}{m})} \right)$.
Here, ${\bf{w}} 
\cdot {\bf{c}}:=c^{(1)}w^{(1)} + \dotsm  + c^{(\frac{n}{m})}w^{(\frac{n}{m})}$.
\end{definition}

As in \cite{bergamaschi2024approaching}, $(C^{(m)})^{\perp} = (C^{\perp})^{(m)}$.  Next, we review the relevant AG code concepts.
Let $X/\mathbb{F}_q$ be a smooth projective curve of genus $g$. Let $D=P_1+\cdots+P_n$
and let $G$ be a divisor on $X$ such that
$P_1,\dots,P_n$ are pairwise distinct $\mathbb F_q$-rational points of $X$, none of which appear in the support of $G$. The Riemann--Roch space of $G$ is defined as
\[
\mathcal{L}(G) = \{ f \in \F_q(X) : (f) + G \ge 0 \} \cup \{0\},
\]
i.e., the set of functions whose poles are bounded by $G$ and $\ell(G) = \dim_{\F_q}\mathcal{L}(G)$. The AG code is defined as
\[
C(D,G)
=
\{(f(P_1),\dots,f(P_n)) : f \in \mathcal{L}(G)\}.\] If $2g-2 < \deg(G) < n$, then $C(D,G)$ is an $[n,k,d]$ code with $k = \ell(G)$, and $d \ge n - \deg(G)$. If $G=mP$ for some positive integer $m$ and $P \in X(\F_q)$, then $C(D,G)$ is called a one-point code. The most common one-point code beyond Reed-Solomon codes are the Hermitian codes, which we will now describe. 
 Let $H/\F_{q^2}$ be the Hermitian curve defined by $y^q + y = x^{q+1}$. Then $H$ has $q^3+1$ $\F_{q^2}$-rational points, including a unique point at infinity, denoted $P_{\infty}$. The one-point Hermitian code $C(D,rP_{\infty})$ with parameters $[q^3, k, \geq q^3 -r]$ over $\F_{q^2}$ is defined as \[C(D, rP_{\infty}) = \{(f(P_1), \ldots, f(P_n)) : f\in \mathcal{L}(rP_{\infty})\} \subseteq \F_{q^2}^n.\]
The dual code of $C(D,rP_{\infty})$ is 
$C(D,rP_{\infty})^{\perp} = C(D, (q^3 - q^2 - q - 2 - r)P_{\infty}).$ Let $I(r) = \{x^iy^j : iq + j(q+1) \leq r\}$. For $0 < r < q^3$, the dimension of $C(D, rP_{\infty})$ is $|I(r)|$. The automorphism group of $H$ satisfies
\[
|\textup{Aut}(H/\F_{q^2})| = (q^3+1)q^3(q^2-1)
\]
providing a source of automorphisms for folding AG codes from the Hermitian curve.
\rmv{We consider automorphisms $\sigma \in \textup{Aut}(H/\F_{q^2})$ of the form
\[
\sigma(x) = \epsilon x + \delta,\qquad
\sigma(y) = \epsilon^{q+1}y + \epsilon\delta^q x + \mu,
\]
where $\epsilon \in \F_{q^2}^*$, $\delta,\mu \in \F_{q^2}$, and $\mu^q + \mu = \delta^{q+1}$. These automorphisms form a subgroup $A$ of $\textup{Aut}(H/\F_{q^2})$. }

The following definition is a special case of that presented in \cite{10.5555/2634074.2634208}.

\begin{definition} Let $n=q^3$ and $\sigma \in \textup{Aut}(H/\F_{q^2})$ be an automorphism of the Hermitian curve such that
$D:=\sum_{i=1}^{\frac{n}{m}} \sum_{j=1}^m P_i^{\sigma^{j}}$ is the sum of the $q^3$ pairwise distinct $\F_{q^2}$-rational points of $H$ other than $P_{\infty}$. Then $C(D, rP_{\infty})^{(\sigma,m)}=$
\[
\left\{ \left(f \left(P_1^{(\sigma,m)} \right), \ldots, f \left(P_N^{ (\sigma,m)} \right) \right)
: f \in \mathcal L(rP_{\infty}) \right\} \rmv{\subseteq (\F_{q^2}^m)^N}\] is associated $m$-folded one-point Hermitian code where
\[ f(P^{(\sigma,m)}):=\begin{bmatrix}
    f(P_1)\\f(P_1^{\sigma})\\
    \vdots \\
    f(P_1^{\sigma^{m-1}})
\end{bmatrix} \in \F_{q^2}^m.\]
\end{definition}

The code $C(D,rP_{\infty})^{(\sigma,m)}$ has parameters $[\frac{n}{m}, \frac{\ell(rP_{\infty})}{m}\geq \frac{n-r}{m}]$. The folded one-point Hermitian codes are defined taking an automorphism from the group above. In order to fold all $q^3$ rational points into equal sized orbits, we require that $m\mid \textup{ord}(\sigma)$ and $m\mid q^3$. Therefore, we will use the automorphisms
\[
\sigma_{\delta,\mu} : (x,y) \mapsto (x+\delta,\, y+\delta^q x + \mu)
\]
where $ \delta, \mu \in \F_{q^2}$ such that $\mu^q + \mu = \delta^{q+1}$, to fold the one-point Hermitian codes for the remainder of the paper.

\section{Folded quantum Hermitian code constructions}\label{results}

In this section, we construct folded quantum Hermitian codes. We begin with the definition of the CSS construction for quantum codes (see \cite{1996PhRvA..54.1098C,PhysRevA.54.4741}). For $x=(x_1,\ldots,x_n)\in \F_q^n$, we write $\ket{x}$ for the computational basis state
\[
\ket{x}=\ket{x_1}\otimes \cdots \otimes \ket{x_n}
\in (\mathbb{C}^q)^{\otimes n}.
\]

Let $C_1, C_2$ be $[n, k_1, d_1]$ and $[n,k_2, d_2]$ codes with $C_2^{\perp} \subseteq C_1$. The associated CSS code is an $[[n,k_1+k_2 - n,\min\{\textup{wt}(C_1\setminus C_2^{\perp}), \textup{wt}(C_2 \setminus C_1^{\perp})\}]]$ code given by
\[\textup{CSS}(C_1, C_2) = \left\{\frac{1}{\sqrt{|C_2^{\perp}|}}\sum_{y \in C_2^{\perp}}\ket{x + y}: x \in C_1 \right\}\]
where $x+y$ is taken coordinatewise and the rate of the CSS code is $R_{\textup{CSS}} = R_{C_1} + R_{C_2} -1.$ \rmv{Let $d^{\star} = \min\{d_1, d_2\}.$ Then, $d^{\star} \leq d$.} Recall that in the classical setting, a code $C \subseteq \Sigma^n$ is said to be $(\tau, L)$ list-decodable if there are at most $L$ codewords of $C$ in any Hamming ball that has radius $\tau \cdot n \in \Sigma^n$ (see \cite{elias1957list,wozencraft1958list}). For AG codes, $\tau$ and $L$ depend on the family of AG codes we are working with. Recall that stabilizer codes are subspaces of $n$-qudit states defined as the joint $+1$ eigenspace of an abelian subgroup of the Pauli group, known as the stabilizer group (see \cite{1997PhDT.......232G}). Error detection is performed by measuring a set of stabilizer generators, and the resulting list of eigenvalues is called the syndrome. This syndrome identifies the presence and type of error up to degeneracy and serves as the quantum analogue of classical parity-check syndromes. The notion of list decoding in the quantum setting can be defined as follows.  

\begin{definition} A stabilizer code $Q$ is said to be $(\tau,L)$ \emph{quantum list-decodable} if, for every possible syndrome $s$, there are at most $L$ logically distinct Pauli operators of relative weight at most $\tau$ that are consistent with the syndrome~$s$.
\end{definition}

 In \cite{bergamaschi2024approaching}, the authors show that after measuring the syndrome, list decoding the stabilizer code is the entirely classical task of computing the list of possible corrections corresponding to the given syndrome.

\begin{theorem} \cite[Lemma 1.5]{bergamaschi2024approaching} Let $C_1, C_2 \subseteq \F_{q}^n$ be two $\F_q$ linear codes that are $(\tau, L)$ list-decodable and let $C_2^{\perp} \subseteq C_1$. Then, there exists a quantum code $Q = \textup{CSS}(C_1, C_2)$ that is $(\tau, L^2)$ quantum list-decodable. If $C_1, C_2$ are efficiently list-decodable, that is, one can find a polynomial time algorithm that outputs the list of codewords, then $Q$ is as well. 
\end{theorem}

In \cite[Lemma 1.6]{bergamaschi2024approaching}, the authors show that if $Q = \textup{CSS}(C_1, C_2)$ where $C_1,C_2 \subseteq \F_{q}^n$, then $Q^{(m)} \subseteq (\F_{q}^m)^{\frac{n}{m}}$ is the CSS code of the two $\F_q$-linear folded classical codes $C_1^{(m)}$ and $C_2^{(m)}$. The next lemma  formalizes an implicit assumption by characterizing exactly when folding commutes with the CSS construction.

\begin{lemma} Let $C_1,C_2$ be linear codes and $C_1^{(m)}$, $C_2^{(m)}$ be their folded versions. We have that $\textup{CSS}(C_1^{(m)}, C_2^{(m)}) = \textup{CSS}(C_1,C_2)^{(m)}$ if and only if $(C_2^{(m)})^{\perp} = (C_2^{\perp})^{(m)}$.
\end{lemma}

\begin{proof}
Folding acts linearly on basis states by mapping each vector
$v\in\mathbb{F}_q^n$ to $v^{(m)}\in\mathbb{F}_{q^m}^{\frac{n}{m}}$.
Thus, $\textup{CSS}(C_1,C_2)^{(m)}$ consists of uniform superpositions over cosets
of $(C_2^\perp)^{(m)}$ inside $C_1^{(m)}$, whereas
$\textup{CSS}(C_1^{(m)},C_2^{(m)})$ consists of uniform superpositions over cosets
of $(C_2^{(m)})^\perp$ inside $C_1^{(m)}$.
These codes coincide if and only if the sets of cosets agree, which holds
precisely when
\(
(C_2^{(m)})^\perp = (C_2^\perp)^{(m)}.
\)
\end{proof}

 We now provide the following construction of folded quantum Hermitian codes.

\begin{construction}\label{foldedH}
Let $H/\mathbb{F}_{q^2}$ be the Hermitian curve, and let
$
C_2 \subseteq C_1 \subseteq \mathbb{F}_{q^2}^n
$
be one-point Hermitian codes of the form
$
C_i := C(D, r_i P_\infty), i=1,2,
$
where $D = \sum_{j=1}^n P_j$ is the sum of all finite rational points and $
0 \le r_2 \le r_1 < n$. Let $\sigma \in \textup{Aut}(H/\F_{q^2})$ be an automorphism of order $m$, which partitions the
evaluation points $\{P_1,\ldots,P_n\}$ into disjoint orbits of size $m$.
\rmv{Let
$
C_i^{\sigma,m} \subseteq (\mathbb{F}_q^m)^{\frac{n}{m}}
$
denote the folded code obtained by grouping codeword coordinates along the
$\sigma$-orbits.} Assume that
$
C_2^{(\sigma,m)} \subseteq C_1^{(\sigma,m)}
$ and
${C_2^{(\sigma,m)}}^{\perp} \subseteq C_1^{(\sigma,m)}$.
Then the folded Hermitian CSS quantum (FQHC) code is defined as $
Q^{(\sigma,m)} :=
\textup{CSS}\bigl(C_1^{(\sigma,m)}, C_2^{(\sigma,m)}\bigr)$.
\end{construction}
We see that the FQHC $Q^{(\sigma,m)} = \textup{CSS}(C_1^{(\sigma,m)}, C_2^{(\sigma,m)})=$ 
\begin{align*}
\Bigg\{
\frac{1}{\sqrt{q^s}}
\sum_{g \in \mathcal{L}(\alpha P_{\infty})}
\bigotimes_{i=1}^{N}
\ket{
(f + g)(P_i^{(\sigma,m)})
}
:\;
f \in \mathcal{L}(r_1 P_{\infty})
\Bigg\}
\end{align*}
where $s:={\frac{\ell(\alpha P_{\infty})}{m}}$.

\rmv{
\begin{construction}\label{foldedH} Let $H = \F_{q^2}(x,y)$, $\sigma \in \textup{Aut}(H/\F_{q^2})$ and let $0 < r_1, r_2 < q^3$. Let $C_i = C(D, r_iP_{\infty})$, $i =1,2$ with $C_2^{\perp} \subseteq C_1$. Let $C_1^{(\sigma,m)}, C_2^{(\sigma, m)}$ be their folded versions with $(C_2^{(\sigma,m)})^{\perp} \subseteq C_1^{(\sigma,m)}$ and $N:=\frac{q^3}{m}$. Let $\alpha = q^3 +q^2 - q - 2 - r_2$. Let $D$ a divisor with $D^{\sigma} = D$. Then the resulting FQHC $Q^{(\sigma,m)} = \textup{CSS}(C_1^{(\sigma,m)}, C_2^{(\sigma,m)})$ is 

\begin{align*}
\Bigg\{
\frac{1}{\sqrt{q^s}}
\sum_{g \in \mathcal{L}(\alpha P_{\infty})}
\bigotimes_{i=1}^{N}
\ket{
(f + g)(P_i^{(\sigma,m)})
}
:\;
f \in \mathcal{L}(r_1 P_{\infty})
\Bigg\}
\end{align*}
where $s:={\frac{\ell(\alpha P_{\infty})}{m}}$.
    
\end{construction}}

\begin{lemma} The FQHC $Q^{(\sigma,m)}$ has parameters \[ \left[\left[\frac{q^3}{m}, \frac{(\ell(r_1P_{\infty}) - \ell(\alpha P_{\infty}))}{m},\geq \left \lceil \frac{\min\{d(C_1), d(C_2)\}}{m}\right \rceil \right] \right]\] where $\alpha = q^3 +q^2 - q - 2 - r_2.$
\end{lemma}

\begin{example} Let $q=2,$ $\epsilon = 1$, $\delta = 0$, $\F_{2^2} = \{0,1,\alpha, \alpha + 1\}$, and $H: x^3 = y^2 + y$, $\mu \in \mathbb{F}_{q^2}$ with $\mu^2 + \mu = \delta^3$. We consider the automorphism $\sigma_{0,\mu}$ defined in the previous section. Let $C_1 = C(D, 4P_{\infty}), C_2 = C(D, 6P_{\infty})$. In this case, $m=2$. Notice that $C_2^{\perp} = C(D, 2P_{\infty}) \subseteq C_1$. 
The $[[4, 1,\geq 2]]$ FQHC has codewords
\begin{align*}
\Bigg\{
\frac{1}{\sqrt{2}}
\sum_{g \in \langle 1,x \rangle }
\bigotimes_{i=1}^{4}
\ket{
\begin{bmatrix}
(f + g)(P_i^{(\sigma,2)})
\end{bmatrix}
}
:\;
f \in \left<1, x, y, x^2 \right>\;
\Bigg\}.
\end{align*}

\end{example}

The choice of automorphism used to fold the one-point Hermitian codes changes the alphabet size and parameters of the classical folded one-point Hermitian codes and the CSS code associated to them. We now have the following example.

\begin{example}
For \(q=4\), set $D = P_1 + \dots + P_{64}$. Let $C_1 = C_2 = C(D, 48P_{\infty})$ both codes of length $64$. Then 
$C_2^{\perp} = C(D, 26P_{\infty}) \subseteq C_1$. For $q=5$, set $D=P_1+\ldots + P_{125}$ and let $C_1 = C_2 = C(D, 95P_{\infty})$ both codes of length $125$, then  
$C_2^{\perp} = C(D, 48P_{\infty}) \subseteq C_1$. For $q=7$, set $D=P_1+\dots + P_{343}$. Let $C_1=C_2 = C(D,245P_\infty)$, then $C_2^{\perp}=C(D,138 P_\infty) \subseteq C_1.$ For $q=8$, set $D = P_1 + \dots + P_{512}$. Let $C_1=C_2 = C(D, 416P_{\infty})$ both codes of length $512$, then $C_2^{\perp} = C(D, 150P_{\infty}) \subseteq C_1$. For $q=16$, set $D= P_1 + \dots + P_{4096}$. Let $C_1=C_2= C(D,3584P_\infty)$, then $C_2^{\perp} = C(D,750P_\infty)\subseteq C_1.$ Consider the automorphisms $\sigma_{0,\mu}$ and $\sigma_{\delta,\mu}$, where
$\mu\neq 0$ for $\sigma_{0,\mu}$ and $\delta\neq 0$ for $\sigma_{\delta,\mu}$. Let $Q^{(\sigma,m)} = \textup{CSS}(C_1^{(\sigma, m)}, C_2^{(\sigma,m)})$. The resulting parameters are shown in Table \ref{tab:fqhc_examples}. 

\begin{table}[h!]
\centering
\renewcommand{\arraystretch}{1.5}
\begin{tabular}{|c|c|c|c|}
\hline
\textbf{Field} & \textbf{$m$} & \textbf{$C_i^{(\sigma,m)}$} & \textbf{$Q^{(\sigma,m)}$} \\ 
\hline
$q=4$ & $2$
& $[32,\tfrac{43}{2},\geq 8]$ 
& $[[32,11,\geq 8]]$ \\ 
\hline
$q=4$ & $4$ 
& $[16,\tfrac{43}{4},\geq 4]$ 
& $[[16,\tfrac{11}{2},\geq 4]]$ \\
\hline
$q=5$ & $5$ 
& $[25,\tfrac{86}{5},\geq 6]$ 
& $[[25,\tfrac{47}{5},\geq 6]]$ \\
\hline
$q=7$ & $7$ 
& $[49, \tfrac{225}{7}, \geq 14]$ & $[[49, \tfrac{107}{7}, \geq 14]]$ \\ 
\hline
$q=8$ & $2$ & $[256, \tfrac{389}{2}, \geq 48]$
& $[[256, 133, \geq 48]]$ \\ 
\hline
$q=8$  & $4$ 
& $[128, \tfrac{389}{4}, \geq 24]$& $[[128, \tfrac{266}{4}, \geq 24]]$ \\ 
\hline
$q=16$ & $2$ & $[2048, \tfrac{3465}{2}, \geq 256]$ & $[[2048, 1417, \geq 256]]$\\
\hline
$q=16$  & $4$ & $[1024, \tfrac{3465}{4}, \geq 128]$ & $[[1024, \tfrac{2834}{4}, \geq 128]]$ \\
\hline
\end{tabular}
\caption{Examples of folded Hermitian codes with respect to $\sigma_{\delta, \mu}$ and the corresponding folded quantum Hermitian codes.}
\label{tab:fqhc_examples}
\end{table}
\end{example}

In comparison, for folded Reed--Solomon codes over $\F_{q^2}$, for $q=4$ we consider
the Reed--Solomon code with parameters $[15,10,6]$. Let $\gamma \in \F_{q^2}^*$ be a primitive element,
and consider the automorphism $\sigma(x) = \gamma x$, which has order $15$
and acts transitively on $\F_{q^2}^*$. Folding is performed by grouping consecutive evaluations along the orbit $\{1, \gamma, \gamma^2, \dots, \gamma^{14}\}.$ Thus, the folding parameter $m$ must divide $15$, so the natural choices
are $m = 3$ or $m = 5$. For example, if $m=3$, we get a code with parameters $[5,10/3,\geq 2]$ and the resulting CSS code has parameters $[[5, 5/3, \geq 2]]$, which has smaller length, dimension and distance for the CSS code.

Unlike the Hermitian case, folded Reed--Solomon codes use a single cyclic orbit of the evaluation set. For Hermitian one-point codes, automorphisms fixing $P_\infty$ decompose the affine rational points into multiple structured orbits, so folding must respect this orbit structure rather than a global cyclic ordering. Thus, Reed--Solomon codes admit simple cyclic foldings, while Hermitian codes exploit richer algebraic and geometric structure and longer length.

The folded Hermitian codes are $(1-R-\epsilon,N^{\mathcal{O}(1/\epsilon^2)})$-list-decodable as shown in \cite[Theorem 3.3]{10.5555/2634074.2634208}. We now discuss the list decodability of the folded quantum Hermitian code constructed in Construction \ref{foldedH}.

\begin{theorem}\label{fqhcl}
Let $l = 1/\epsilon^2$ be a square prime power, $q = l^2 = 1/\epsilon^4$, $m = l$ and $\epsilon > 0$. 
The folded quantum Hermitian code (FQHC) is 
$\left(\frac{1 - R_{FQHC}}{2} - \epsilon,\; N^{\mathcal{O}(1/\epsilon^2)} \right)$-list-decodable,
where $R_{FQHC} = R_{C_1^{(\sigma,m)}} + R_{C_2^{(\sigma,m)}} - 1.$
\end{theorem}

\begin{proof}
By \cite[Lemma 1.6]{bergamaschi2024approaching}, the folded quantum Hermitian code is a CSS code constructed from two folded Hermitian codes $C_1^{(\sigma,m)}$ and $C_2^{(\sigma,m)}$. Folded Hermitian codes are folded AG codes and are $(1 - R - \epsilon,N^{\mathcal{O}(1/\epsilon^2)})$-list-decodable by \cite[Theorem 3.3]{10.5555/2634074.2634208} for the choice of $l$ a squared prime power with $l = 1/\epsilon^2, m = l$ and $q = l^2$.
Now we have that 
\[
\tau_{C_1^{(\sigma,m)}} = 1 - R_{C_1^{(\sigma,m)}} - \epsilon,
\tau_{C_2^{(\sigma,m)}} = 1 - R_{C_2^{(\sigma,m)}} - \epsilon.
\]

By Lemma~1.5 of \cite{bergamaschi2024approaching},
the CSS construction yields a quantum code whose
list-decoding radius equals
\[
\tau_{FQHC}
=
\min \left\{
\tau_{C_1^{(\sigma,m)}},
\tau_{C_2^{(\sigma,m)}}
\right\}.
\]

We now determine the choice of classical rates that maximizes the quantum decoding radius for fixed $R_{FQHC} = R_{C_1^{(\sigma,m)}} + R_{(C_2^{(\sigma,m)}} - 1$.
Since
\[
\tau_{FQHC}
=
\min\left\{1 - R_{C_1^{(\sigma,m)}},\, 1 - R_{C_2^{(\sigma,m)}}\right\}
- \epsilon,
\]
maximizing $\tau_{FQHC}$ is equivalent to minimizing $\max\{R_{C_1^{(\sigma,m)}},\, R_{C_2^{(\sigma,m)}}\}$
subject to the constraint $R_{C_1^{(\sigma,m)}} + R_{C_2^{(\sigma,m)}} = 1 + R_{FQHC}.$ For fixed rate, the quantity $\max\{R_{C_1^{(\sigma,m)}},\, R_{C_2^{(\sigma,m)}}\}$ is minimized precisely when $R_{C_1^{(\sigma,m)}} = R_{C_2^{(\sigma,m)}}$. Any asymmetric choice strictly increases the larger rate and therefore strictly decreases the decoding radius. Hence, the optimal choice is $R_{C_1^{(\sigma,m)}} = R_{C_2^{(\sigma,m)}}
= \frac{1 + R_{FQHC}}{2}.$ Substituting into the expression for the decoding radius gives
\[
\tau_{FQHC}
=
1 - \frac{1 + R_{FQHC}}{2} - \epsilon
=
\frac{1 - R_{FQHC}}{2} - \epsilon.
\] 

Since each folded Hermitian code has list size $L= N^{\mathcal{O}(1/\epsilon^2)},$ by \cite[Lemma 1.5]{bergamaschi2024approaching}, the list size of the resulting FQHC is also $N^{\mathcal{O}(1/\epsilon^2)}$. Therefore, the FQHC is $(\frac{1- R_{FQHC}}{2} - \epsilon, N^{\mathcal{O}(1/\epsilon^2)})$-list decodable as desired. 
\end{proof}

The previous result shows that the folded quantum Hermitian codes are list-decodable codes that tolerate errors up to the quantum Singleton bound. We now have a result combining this with quantum list decoding. Some parts are inspired by the proof of \cite[Theorem 1.8]{bergamaschi2024approaching} which appears in the full arXiv version of the paper. 

\begin{proposition}  Let $\epsilon > 0$. Let $l = 1/\epsilon^2$ be a square prime power, $q = l^2 = 1/\epsilon^4$, $m = l$ and $\epsilon > 0$. Then, there exists a family of folded quantum Hermitian codes $Q^{(\sigma,m)}$ of rate $R_{FQHC}$ and relative distance $\delta \geq \frac{1-R_{FQHC} - \epsilon}{2}$ with alphabet size $(1/\epsilon)^{\mathcal{O}(1/\epsilon^2)}$ that are $(\frac{1-R_{FQHC}}{2} - \epsilon,N^{\mathcal{O}(1/\epsilon^2)})$-list-decodable, where $R_{FQHC} = R_{C_1^{(\sigma,m)}} + R_{C_2^{(\sigma,m)}} - 1$ in time $N^{\mathcal{O}(1/\epsilon^2)}$. Furthermore, there exists a deterministic quantum list-decoding algorithm that
outputs a polynomial-size list containing the correct Pauli error in time
$N^{\mathcal O(1/\epsilon^2)}$.
    
\end{proposition}

\begin{proof} Let $C_1, C_2$ be two classical one-point folded Hermitian codes with $(C_2^{(\sigma,m)})^{\perp} = (C_2^{\perp})^{(\sigma,m)} \subseteq C_1^{(\sigma,m)}$. Let $Q^{(\sigma,m)} = \textup{CSS}(C_1^{(\sigma,m)}, C_2^{(\sigma,m)})$, constructed as in Construction \ref{foldedH}. By construction, $Q^{(\sigma,m)}$ has rate $R_{FQHC} = R_{C_1^{(\sigma,m)}} + R_{C_2^{(\sigma,m)}} - 1$ and relative distance $\delta \geq \frac{1-R_{FQHC} - \epsilon}{2}$. Since folded Hermitian codes in \cite{10.5555/2634074.2634208} are explicit and 
$(1-R-\epsilon, N^{\mathcal{O}(1/\epsilon^2)})$-list-decodable with runtime 
$N^{\mathcal{O}(1/\epsilon^2)}$, it follows from Theorem~\ref{fqhcl} and 
\cite[Lemma~1.5]{bergamaschi2024approaching} that the CSS code 
$Q^{(\sigma,m)}$ is \[\left(\frac{1-R_{FQHC}}{2}-\epsilon, N^{\mathcal{O}(1/\epsilon^2)}\right)\] list-decodable. Hence, $Q^{(\sigma,m)}$ has rate 
$R_{FQHC}$, relative distance $\delta \ge \frac{1-R_{FQHC}-\epsilon}{2}$, and constant
alphabet size $(1/\epsilon)^{\mathcal{O}(1/\epsilon^2)}$. 

It remains to show that there exists a deterministic decoding algorithm and to describe it. By \cite[Lemma 1.5]{bergamaschi2024approaching}, the quantum list decoding of the CSS code we constructed reduces to classical list decoding of $C_1^{(\sigma,m)}$ and $C_2^{(\sigma,m)}$. The quantum list decoder first measures all stabilizers to obtain classical syndromes, which provide linear constraints on the errors from $C_1^{(\sigma,m)}$ and $C_2^{(\sigma,m)}$. It then runs the classical list decoder from
\cite{10.5555/2634074.2634208} separately on each syndrome to obtain
candidate lists $L_X$ and $L_Z$. Finally, the classical list decoder then combines these candidate classical error patterns into a list of candidate Pauli operators by pairing each candidate from $C_1^{(\sigma,m)}$ with each candidate from $C_2^{(\sigma,m)}$. Pairing the two lists
produces a quantum list of candidate Pauli errors containing the true Pauli error. Since each classical list has size at most
$N^{\mathcal O(1/\epsilon^2)}$ and constant factors in the exponent are absorbed, the resulting list size of the folded quantum Hermitian code is at
most
\[
|L_X||L_Z|
\le
N^{\mathcal O(1/\epsilon^2)}
\cdot
N^{\mathcal O(1/\epsilon^2)}
=
N^{\mathcal O(1/\epsilon^2)}.
\]  Each classical decoding step runs in time
$N^{\mathcal O(1/\epsilon^2)}$. Therefore, the overall decoding procedure is
deterministic and runs in time $N^{\mathcal O(1/\epsilon^2)}$. This
establishes the existence of a deterministic quantum list-decoding
algorithm that outputs a polynomial-size list containing the correct
Pauli error. 
\end{proof}

\begin{remark} Folded quantum Hermitian codes have smaller alphabet sizes than folded quantum Reed-Solomon codes constructed in \cite[Theorem 1.8]{bergamaschi2024approaching} using quantum distance amplification. 
    To see this, let $\epsilon > 0$ and let $C_1^{(\sigma,m)}$ and $C_2^{(\sigma,m)}$ be any two folded Hermitian codes. Let $FQHC = \textup{CSS}(C_1^{(\sigma,m)},C_2^{(\sigma,m)})$ be a folded quantum Hermitian code. Then $Q$ has alphabet size $(1/\epsilon)^{\mathcal{O}(1/\epsilon^2)}$ (since they inherit the alphabet size from the two classical folded Hermitian codes). Let FQRS be a folded quantum Reed-Solomon code as in \cite[Theorem 1.8]{bergamaschi2024approaching}. The code FQRS has alphabet size $2^{\mathcal{O}(1/\epsilon^5)}$. Comparing $2^{\mathcal{O}(1/\epsilon^5)}$ to $(1/\epsilon)^{\mathcal{O}(1/\epsilon^2)}$, for any $\epsilon,$ $2^{\mathcal{O}(1/\epsilon^5)}$ is always larger than $(1/\epsilon)^{\mathcal{O}(1/\epsilon^2)}.$ Therefore, the alphabet size of an FQHC code without distance amplification is always smaller than a FQRS code with distance amplification. This demonstrates that in constructing folded quantum Hermitian codes, the alphabet size is constant, therefore quantum distance amplification is not needed. 
\end{remark}

We now provide a folded quantum code design based on the entanglement assisted construction presented in \cite{Galindo2021CorrectionTE}. This construction allows any two folded Hermitian codes to define a folded quantum Hermitian code without requiring $C_2^{\perp}\subseteq C_1$,  
$(C_2^{(\sigma,m)})^{\perp}\subseteq C_1^{(\sigma,m)},$
and the assumption that $(C^{(\sigma,m)})^{\perp}=(C^{\perp})^{(\sigma,m)}$ as in the previous section.
Here, $c$ denotes the number of required ebits, that is, the dimension of the entanglement. 

\begin{proposition} Let $C_1, C_2$ be $[n,k_1,d_1]$ and $[n,k_2,d_2]$ Hermitian codes and let $\widetilde{C}_1 = C_1^{(\sigma,m)}$, $\widetilde{C}_2=C_2^{(\sigma,m)}$ be the respective folded Hermitian codes. Then, there exists $[[\frac{n}{m},k_{EA},d_{EA}:c]]$ quantum code $Q$ with $c = \dim{\widetilde{C}_1} - \dim(\widetilde{C}_1\cap (\widetilde{C}_2)^{\perp}), k_{EA} = \frac{n}{m} - \dim(\widetilde{C}_1) - \dim(\widetilde{C}_2) + c $ and
\[ d_{EA} \geq \begin{cases}
    \min\{\textup{wt}((\widetilde{C}_1)^{\perp}), \textup{wt}((\widetilde{C}_2)^{\perp})\} &\textnormal{if }   (\widetilde{C}_1)^{\perp} \subseteq \widetilde{C}_2 \\
    \min\{\textup{w}_1, \textup{w}_2\} & \textnormal{otherwise}
\end{cases}\]
where $\textup{w}_1:=\textup{wt}((\widetilde{C}_1)^{\perp}\setminus (\widetilde{C}_2\cap (\widetilde{C}_1)^{\perp}))$ and \\$\textup{w}_2:= \textup{wt}((\widetilde{C}_2)^{\perp}\setminus ((\widetilde{C}_2)^{\perp}\cap \widetilde{C}_1)$.

\rmv{
\[ d_{EA} \geq \begin{cases}
    \min\{\textup{wt}((\widetilde{C}_1)^{\perp}), \textup{wt}((\widetilde{C}_2)^{\perp})\} \hbox{ if }   (\widetilde{C}_1)^{\perp} \subseteq \widetilde{C}_2 \\
    \min\{\textup{wt}((\widetilde{C}_1)^{\perp}\setminus (\widetilde{C}_2\cap (\widetilde{C}_1)^{\perp})),\\ \textup{wt}((\widetilde{C}_2)^{\perp}\setminus ((\widetilde{C}_2)^{\perp}\cap \widetilde{C}_1)\} \hspace{2ex} \hbox{otherwise}.
\end{cases}\]}    
\end{proposition}

\section{Conclusion}\label{conclusion}

In this paper, we constructed folded quantum Hermitian codes, showing that they are list-decodable tolerating errors up to the quantum Singleton bound. Our
construction achieves this over alphabet size
$(1/\epsilon)^{\mathcal O(1/\epsilon^2)}$, improving upon folded quantum
Reed--Solomon codes in terms of alphabet efficiency. Moreover, we
demonstrated that costly techniques such as distance amplification are
not required to obtain near-optimal decoding radius. It remains open to construct explicit folded quantum AG codes from other curves of positive genus and to study the trade-offs between rate,
distance, and alphabet size in these settings. \rmv{In particular, it would be
interesting to determine whether similar constructions can achieve
near-optimal list-decoding radius with smaller alphabets or improved
parameter regimes.} It also remains open explore other folded quantum code constructions beyond using the CSS and entanglement-assisted quantum code constructions.

\bibliographystyle{IEEEtran}
\bibliography{foldedAGpaper}

\end{document}